\documentclass[journal,a4paper]{IEEEtran}
\usepackage{amsthm}
\usepackage{graphicx}
\usepackage{subcaption}

\usepackage{cite}
\usepackage{xcolor}
\usepackage{amsmath}
\usepackage{amssymb}
\usepackage{mathtools}

\usepackage{algorithm}
\usepackage{algorithmic}
\makeatletter
\def\BState{\State\hskip-\ALG@thistlm}
\makeatother
\usepackage{fixltx2e}

\usepackage{stfloats}
\usepackage{lipsum}
\usepackage[margin=0.75in]{geometry}
\newgeometry{top=122pt, left=54pt, right=37pt, bottom=54pt}

\begin{document}
\setlength{\abovecaptionskip}{1pt}
\setlength{\belowcaptionskip}{-15pt}
\setlength{\abovedisplayskip}{0.7pt}
\setlength{\belowdisplayskip}{0.7pt}
\setlength{\abovedisplayshortskip}{0.5pt}
\setlength{\belowdisplayshortskip}{0.5pt}
\setlength{\parskip}{0.5pt}
\setlength{\textfloatsep}{0.9pt}
\setlength{\floatsep}{.5pt}
%
\title{Smart Charging Benefits in Autonomous Mobility on Demand Systems}

\author{\IEEEauthorblockN{Berkay Turan \quad}
 \and
\IEEEauthorblockN{Nathaniel Tucker \quad}
\and
\IEEEauthorblockN{Mahnoosh Alizadeh}
\vspace*{-0.85cm}
}


%


\maketitle

\begin{abstract}
In this paper, we study the potential benefits from smart charging for a fleet of electric vehicles (EVs) providing autonomous mobility-on-demand (AMoD) services. We first consider a profit-maximizing platform operator who makes decisions for routing, charging, rebalancing, and pricing for rides based on a  network flow model. Clearly, each of these decisions directly influence the fleet's smart charging potential; however, it is not possible to directly characterize the effects of various system parameters on smart charging under a classical network flow model. As such, we propose a  modeling variation that allows us to decouple the charging and routing problems faced by the operator. This variation allows us to provide closed-form mathematical expressions relating the charging costs to the maximum battery capacity of the vehicles as well as the fleet operational costs. We show that investing in larger battery capacities and operating more vehicles for rebalancing reduces the charging costs, while increasing the fleet operational costs. Hence, we study the trade-off  the  operator faces, analyze the minimum cost fleet charging strategy, and provide numerical results illustrating the smart charging benefits to the  operator.
\end{abstract}

%
\IEEEpeerreviewmaketitle
\newtheorem{proposition}{Proposition}
\newtheorem{corollary}{Corollary}[proposition]
\newtheorem{theorem}{Theorem}
\newtheorem{lemma}{Lemma}
\makeatletter
\def\blfootnote{\xdef\@thefnmark{}\@footnotetext}
\makeatother

\blfootnote{
\hspace{-12pt}
This work was supported by the NSF Grants 1837125, 1847096.\\
B. Turan, N. Tucker, and M. Alizadeh are with the Dept. Electrical and Computer Engineering, University of California, Santa Barbara.}

\section{Introduction}
The increasing popularity of mobility-on-demand platforms, the rapid developments in autonomous driving technology, and the increasing adoption rate of EVs are disruptive technologies that are extensively altering society's perspective of urban mobility.  Given this, the vision of an electric and autonomous mobility-on-demand (AMoD) fleet serving urban customers' mobility needs is gaining traction within the transportation industry, with multiple companies now heavily investing in AMoD technology \cite{companies}. 

In conjunction with society's interest in AMoD technologies, there is extensive literature emerging that studies the different aspects of AMoD systems. The potential impact of shared mobility services on daily urban mobility \cite{simulation5}, the analysis of AMoD systems with realistic demand  \cite{simulation4}, autonomous vehicle behavior in existing traffic models \cite{simulation2}, and rebalancing algorithms \cite{simulation1} have been investigated using simulation frameworks. The interplay between AMoD and public transport has been studied in \cite{simulation6} and \cite{publictransit1}. On the modeling side, queueing theoretical models capture the stochasticity of the customers  \cite{queuetheoretical2}, while network flow models efficiently optimize the fleet control in a static setting \cite{networkflow1}. Owing to their simplicity, network flow based formulations are commonly used for algorithmic control of routing and rebalancing  in a receding-horizon fashion \cite{recedinghorizon1}, and to control congestion effects \cite{congestion1}.

In addition to AMoD technology, the transportation sector is looking to increase utilization of electric vehicles (EVs) for consumers, companies, and fleet operations. Specifically, EV employment in mobility-on-demand (MoD) systems such as taxi environments has been studied in \cite{evtaxi1} and \cite{evtaxi2}. The authors of \cite{scheduling1} study scheduling algorithms for assigning MoD EVs to trips. To address the issue of EVs' need to perform in-route charging, \cite{inroutecharging} proposes a routing scheme that aims to reduce overall delays. As EVs could also be autonomous, the authors of \cite{agentbased} study an agent-based model to simulate the operations of an AMoD fleet of EVs under various vehicle and infrastructure scenarios. Paper \cite{nate} proposes an online charge scheduling algorithm for EVs providing AMoD services. Possible issues such as communication delays and instability during charging process arising from occupying EVs in an AMoD system are investigated in \cite{fogbased}. As these fleets can be used to assist other means of transportation, \cite{trainconnection} discusses the potential of using AMoD as a last mile connection of train trips. 
Additionally, the authors of \cite{powergrid1} analyze the interaction between EV AMoD systems and power grid due to the charging requirements of EVs.

In this paper, our main goal is to quantify the decreased charging costs from utilizing smart charging as an EV AMoD fleet transports customers between various locations. Specifically, smart charging refers to the practice of charging EVs opportunistically at times and locations at which electricity is inexpensive and the power grid is under less stress. We adopt a model that considers a profit-maximizing AMoD platform that is transporting customers over a static  and simplified network. We assume that the platform operator optimizes vehicle charging and rebalancing decisions, as well as customer payments for rides, by  considering the diversity of electricity prices at different network nodes. Exploiting the diversity in electricity prices leads to an EV smart charging plan in the context of AMoD systems. Moreover, our goal is to study the effects of various system parameters on the cost savings that smart charging can provide.  

We adopt an abstract network flow-based formulation and focus on the impacts of two critical factors on the charging costs: (i) the EV battery capacity; (ii) the per-vehicle operational costs (and implicitly, the fleet size). Accordingly, we analyze optimal routing, pricing, charging, and rebalancing strategies for the flow-based formulation. Then, in order to quantify the importance of battery capacity and  operational costs on the smart charging potential of AMoD fleets, we adopt a  modeling variation  that decouples charging decisions from routing decisions, hence allowing us to obtain closed-form expressions for the trade-offs that are of interest to us. While it is evident that the closed-form relationships are only valid for our abstract model,  they highlight important design choices for any AMoD system utilizing EVs.


\textit{Organization}: The remainder of the paper is structured as follows: Section \ref{systemmodel} presents the system model and describes the platform operator's optimization problem. In Section \ref{netflow}, we formulate the optimization problem using a network flow approach and discuss the effects of fleet operational and charging costs on profits. Section \ref{randomprices} proposes a modeling variation in order to mathematically characterize smart charging benefits. Section \ref{numerical} presents numerical results quantifying the smart charging benefits.
\section{System Model}\label{systemmodel}

{\it Network and Demand Models:}
We consider a fleet of AMoD EVs operating within a transportation network that is a fully connected graph consisting of $\mathcal M = \{1,\ldots,m\}$ equidistant nodes that can each serve as a trip origin or destination\footnote{\label{note1}Most of our results can be extended to the more general case with nodes being geographically distributed on a network and hence different trips taking different amounts of energy and time. For brevity of notation we use equidistant nodes in this paper.}. We adopt the static model studied in \cite{ridesharing} for the customers' transportation demand. We assume that potential customers arrive at node $i$ at a  rate  of $\theta_i$  per period. The routing matrix $A = [\alpha_{ij}]_{i,j \in \mathcal M}$ defines the fractions $\alpha_{ij}$ of riders at node $i$ who wish to go to node $j$, with $\alpha_{ii}=0$, $\alpha_{ij}\geq 0$, and $\sum_{j\in \mathcal M}\alpha_{ij}=1$. Moreover, we assume that these riders are heterogeneous in terms of their willingness to pay. In particular, if the price for receiving a ride from node $i$   is set to $\ell_i$, the induced demand for rides from $i$ to $j$ at each time period is given by $\Lambda_{ij}=\theta_i\alpha_{ij}(1-F(\ell_i))$, where $F(\cdot)$ is the cumulative distribution of riders' willingness to pay with a support of $[0,\ell_{\max}]$. We note that the price of rides is only dependent on their origin; however, an extension to origin-destination (O-D) based prices is straightforward.

{\it Vehicle Model:} 
To capture the effect of trip demand and the associated charging, routing, and rebalancing decisions  on the fleet size, we assume that each autonomous vehicle in the fleet has a per period operational cost of $\beta$. As such, we make no explicit assumption on fleet size; rather, our cost model implicitly optimizes the fleet size given the system parameters. Furthermore, as the vehicles are electric, they have to sustain charge in order to operate. We assume there is a charging station placed at each node $m\in\mathcal M$. To charge at node $i$, the operator pays a price of electricity $p_i$ per unit of energy.  
We assume that all EVs in the fleet have a battery capacity denoted as $v_{\max}\in \mathbb Z^+$; therefore, each EV has a discrete battery energy level $v \in \mathcal V$, where $\mathcal V = \{v\in \mathbb{N}| 0\leq v \leq v_{\max}\}$. In our discrete-time model, we assume each vehicle takes one period to charge one unit of energy. Given that all O-D pairs are considered to be equidistant, we assume each trip takes $\tau$ periods of time to complete and consumes one unit of energy\footnotemark[1]. 

{\it Rebalancing:} 
In addition to routing and charging the vehicles, the fleet operator can also utilize vehicles for rebalancing. Specifically, these are vehicles that are traveling between different nodes in the network without carrying passengers. Rebalancing vehicles are required for the platform to serve the induced outgoing demand at a node if said demand exceeds the incoming trip demand with that node as the destination. Moreover, as we emphasize in the following sections, rebalancing trips can also be useful for lowering the platform's charging costs. Thus, in our model, even with a completely balanced trip pattern (i.e., the induced demand being equal to the incoming demand at each node), rebalancing vehicles may still be employed by the operator.

{\it Platform Operator's Problem:} We consider a profit-maximizing AMoD operator that manages a  fleet of EVs that make trips to provide transportation services to customers. The operator's goal is to maximize profits by 1) setting prices for rides and hence managing customer demand at each node; 2) optimally operating the AMoD fleet (i.e., charging, routing, and rebalancing) to minimize operational and charging costs.

\section{Network Flow Formulation and Marginal Prices}\label{netflow}
\subsection{Network Flow Model}

In this section, we approach the platform's optimization problem via a network flow model. 
Specifically, let $\ell_i$ be the price for rides originating from node $i$, $x_i^v$   the number of vehicles with battery energy level $v$ charging at node $i$, $x_{ij}^v$  the number of vehicles starting with a battery energy level $v$ and transporting a passenger from node $i$ to $j$, and $r_{ij}^v$  the number of rebalancing vehicles starting with a battery energy level $v$ and making a trip from node $i$ to $j$. The platform operator aims to set $\ell_i$, $x_i^v$, $x_{ij}^v$, and $r_{ij}^v$ in order to maximize profits $P$.
Namely, the operator's problem can be stated as:
\begin{equation}
\label{eq:flowoptimization}
\begin{aligned}
&\underset{x_i^v,x_{ij}^v,r^v_{ij},\ell_i}{\text{max}}
& &\sum_{i=1}^m\ell_i\theta_i(1-F(\ell_i))-\sum_{i=1}^m\sum_{v=0}^{v_{\max}-1} (\beta+p_i) x_i^v\\& & &-\tau\beta \sum_{i=1}^m\sum_{j=1}^m\sum_{v=1}^{v_{\max}} (r_{ij}^v+x_{ij}^v) \\
& \text{subject to}
& & \sum_{v=1}^{v_{\max}}x_{ij}^v = \theta_i(1-F(\ell_i))\alpha_{ij}\quad\forall i,j \in \mathcal M,\\
& & & x_i^v+\sum_{j=1}^m(x_{ij}^v+r_{ij}^v)=\\& & &x_i^{v-1}+\sum_{j=1}^m(x_{ji}^{v+1}+r_{ji}^{v+1})\quad\forall i\in\mathcal M,\;\forall v\in\mathcal V,\\
& & & x_i^v\geq 0, \; x_{ij}^v\geq 0,\; r_{ij}^v\geq 0\; ~\forall i,j\in \mathcal M,\;\forall v\in\mathcal V.
\end{aligned}
\end{equation}
The first term in the objective function in \eqref{eq:flowoptimization} accounts for the aggregate revenue the platform generates by providing rides for $\theta_i(1-F(\ell_i))$ number of riders with a price of $\ell_i$. The second term is the operational and charging costs incurred by the charging vehicles, and the last term is the operational costs of the trip-making vehicles (including rebalancing trips). The first constraint requires the platform to serve all the induced demand between any two nodes $i$ and $j$. We will refer to this as a the {\it demand satisfaction constraint}. The second constraint is the  flow balance constraint for each node and each battery energy level  (For brevity, we have specified the constraints for all $v \in\mathcal V$. The variables with superscripts outside the set are equal to zero).

The optimization problem in \eqref{eq:flowoptimization} is non-convex for a general $F(\cdot)$. Nonetheless, when the platform's profits are affine in the induced demand $\theta_i(1-F(\cdot))$, it can be rewritten as a convex optimization problem. Hence, we assume that the rider's willingness to pay is uniformly distributed in $[0,\ell_{\max}]$, i.e., $F(\ell_i)=\frac{\ell_i}{\ell_{\max}}$.

\subsection{Marginal Pricing}
The optimal prices $\ell_i^*$ are related to the operational and charging costs associated with making a trip out of node $i$. The next proposition highlights this relationship.

\begin{proposition}\label{prop:networkflow}
Let $\lambda_{ij}^*$ be optimal the dual variable corresponding to the demand satisfaction constraint for trips originating at node $i$ and ending in node $j$. The optimal prices $\ell_i^*$ for rides originating at node $i$ are:
\begin{equation}
\label{eq:optimalprices}
    \ell_i^*=\frac{\ell_{\max}+\sum_{j=1}^m\lambda_{ij}^*\alpha_{ij}}{2}.
\end{equation}
These prices can be upper bounded by:
\begin{equation}\label{eq:bounprices}
     \ell_i^*\leq\frac{\ell_{\max}+p_i+\sum_{j=1}^m\alpha_{ij}p_j+(2+2\tau)\beta}{2}.
\end{equation}
Moreover, with these optimal prices $\ell_i^*$, the profits generated per period is:
\begin{equation}
    \label{eq:profits}
    P=\sum_{i=1}^m\frac{\theta_i}{\ell_{\max}}(\ell_{\max}-\ell_i^*)^2.
\end{equation}
\end{proposition}
The dual variables $\lambda_{ij}^*$, could be interpreted as the cost of providing a single ride between $i$ and $j$ to the platform. In the worst case scenario, every single requested ride from node $i$ requires rebalancing and charging both at the origin and the destination. Hence the upper bound  \eqref{eq:bounprices} includes the price of electricity at the trip origin (to charge the rebalancing vehicle), the average price of electricity at the destination and the operational cost of 4 vehicles, 2 of which are used for trips and 2 for charging.

\subsection{Smart Charging Benefits}
The cost $\lambda_{ij}^*$ of providing a single ride between nodes $i,j \in \mathcal M$ is fundamental to the operations of the AMoD system. Consider the results presented in Proposition \ref{prop:networkflow}.
We can observe that the platform profit $P$ is lowered as the  additional cost term $\sum_{j=1}^m\lambda_{ij}^*\alpha_{ij}$ in \eqref{eq:optimalprices} increases.   This additional term $\sum_{j=1}^m\lambda_{ij}^*\alpha_{ij}$, which is the average marginal cost of a single ride out of node $i$,  could  be  interpreted  as taxes applied on products, which is shared among the supplier and  the  consumer  in  a  basic  supply-demand  setting.  In  the AMoD system, the cost is shared equally among the platform operator and  the  riders, which results in a decrease in both profits and consumer surplus. To decrease this loss, the platform operator acts in order to decrease the total cost of operation (i.e., charging and fleet operational costs) via smart charging and routing strategies. Our goal is to specifically study how   smart charging strategies can aid the operator in decreasing the costs of rides. The potential of smart charging strategies for reducing costs clearly depends on the battery capacity $v_{\max}$ and the operational cost parameter $\beta$.

Let us elaborate further. A smart charging strategy allows the vehicles to avoid charging at expensive nodes and charge as much as they can once they arrive at a cheap node. The lower the battery capacity $v_{\max}$ is, the less likely it is for a vehicle to visit nodes with cheaper electricity prices before running out of charge. Alternatively, a large enough battery capacity $v_{\max}$ allows the operator to solely charge the vehicles at cheap nodes, resulting in a low electricity cost. In a similar manner, a rebalancing trip to a cheaper node (As mentioned in Section II, rebalancing can also be done solely for charging purposes.) could decrease the total   costs, even though it increases the fleet operational costs. As an example, consider the following setting: Let's assume that, unless rebalancing is allowed, the optimal strategy for a vehicle at node $i$ with $p_i=3$ is to charge for one unit of energy. There is another node $j$, with $p_j=0.5$. Let $\beta=0.02$ and $\tau=10$. Then, instead of paying $p_i+\beta=3.02$ to charge for a single unit of energy at node $i$, this vehicle should visit node $j$, charge for $3$ units and then come back to node $i$ for a total cost of $3\beta+2\tau\beta+3p_j=1.96$. Clearly, the profitability of such rebalancing trips depends on the value of $\beta$.

This network flow model accounts for all the phenomena   mentioned above. Yet it is not possible to explicitly characterize the benefits of employing smart charging  strategies alone on reducing the cost of rides between nodes $i$ and $j$. This is because we cannot explicitly state the relationship between the dual   multipliers $\lambda_{ij}^*$ of the optimization problem \eqref{eq:flowoptimization} with   the demand's willingness to pay characterized by $F(\cdot)$, the potential demand $\theta_i$,  the routing matrix $A$, the electricity prices, as well as our parameters of interest, $v_{\max}$ and $\beta$.
Hence, while we will numerically study this interconnection and its effects on the platform's profit in Section \ref{numerical},   we would like to propose a variation of the same  flow model that enables us to decouple the effects of the network parameters $\theta_i$, $F(\cdot)$ and $A$ from the optimal charging strategy and allows us to
provide explicit relationships between the cost-savings due to smart charging and our parameters of interest, namely the electricity price diversity in the network, battery capacity $v_{\max}$ and the fleet operational cost parameter $\beta$.
\section{Smart Charging Benefits with Random Prices}\label{randomprices}
 In this section, we propose a   variation of the network flow model that allows us to explicitly characterize the relationship between optimal charging cost incurred for each individual trip as a function of the vehicles' battery capacity $v_{\max}$ and the fleet operational cost parameter $\beta$. This can highlight an important planning trade-off that an AMoD operator faces: by investing in a larger fleet or in vehicles with larger battery capacities, the day-to-day costs of the operator can decrease.

 Specifically, from now on, we will assume that for the purposes of planning, the operator considers  the prices of electricity  that a vehicle can see at all nodes except the current node they are located at to be iid random variables sampled from a continuous distribution $f_P(p)$ with support   $[p_{\min},p_{\max}]$. However, the price of electricity at the current node will be considered known and constant for the duration of charge events once it is observed. This can be justified if the prices are strict sense stationary random processes with little correlation given time lags of order $\tau$. 
 Such random price models can have real-world applications given the introduction of high levels of renewable energy in the power grid, which makes electricity prices harder to forecast on a day by day basis.
 
 Hence, while we retain all the elements of our static flow model, we assume that the electricity prices at the destination nodes of all current trips is unknown to the operator at the time the optimization problem \eqref{eq:flowoptimization} is solved.   This assumption effectively  decouples the network operator's decision problem into two independent components: 
 \begin{enumerate}
     \item that of deciding whether to charge a vehicle with battery energy level $v$ if it is currently located at a node with electricity price $p$. By solving for the optimal charging strategy in this stochastic setting, we can characterize  the average charging cost $p_{avg}$ that must be paid for each trip by each vehicle; 
     \item that of deciding the optimal price to charge for rides at each node $i$ and the rebalancing trips performed in a non-electric AMoD system to ensure network balance between supply and demand at each node. This is essentially equivalent to solving problem \eqref{eq:flowoptimization} with all $p_i$'s set to a constant $p_{avg}$ given by the first problem.
 \end{enumerate}

 This form of stochasticity  in the prices  mean that   the vehicles' charging strategy would now solely depend on their current state of charge (SoC) and current price tuple $(v,p)$ (as opposed to the network flow model). This is a natural consequence of the fact that charging decisions are entirely independent of where the vehicle will be sent to next,   since the prices of electricity at all possible destinations the vehicle might be routed to is unknown to the operator. As a result, the solution of the first problem yields an average charging cost for every vehicle in the network.

\subsection{The Optimal Charging Strategy under Random Prices}
In this section, we develop an optimal charging strategy under the random price model. The decision of whether to charge or not solely depends on the vehicle's current SoC $v$ and the electricity price observed at the current node $p$. Hence, we define the optimal charging policy  $\mu$ as a collection of sets $\mathcal P_v,~v = 0, \ldots, v_{\max}-1$. The  prices $p \in \mathcal P_v$  are those at which it is optimal for a vehicle with battery energy level $v$ to charge for one unit.  If the price of the current node does not fall in  $\mathcal P_v$, the vehicle will not charge   and will instead travel to the next node (as long as $v \geq 1$).

Our goal is to determine the policy $\mu$ that minimizes average charging cost and subsequently, use this analysis to study the effect of the vehicles' battery capacity and the fleet operational cost parameter $\beta$ on the average charging cost.  
\begin{lemma}
Under the optimal policy, we have:\begin{equation}\mathcal P_v = \{p|  p \in [p_{\min},C_v(\mu)] \},\end{equation}
where $C_v(\mu)$ is the expected price of the next unit of energy under the policy $\mu$ if leaving the current node with a battery energy level $v$. 
\end{lemma}
The proof simply follows from the Bellman equation considering the fact that $C_0(\mu) = \infty$ under any policy $\mu$. Essentially, to make the charging decision, a comparison of the price of electricity at the current node and the expected price to be paid for the next unit of energy if the vehicle leaves the current node has to be made. If the current price is less than the expected price, the decision is to charge. Else, the vehicle does not charge and leaves the node. 
Hence,
\begin{multline}
\label{Cv}
        C_v(\mu)=P(p< C_{v-1}(\mu))\mathbb{E}[p|p< C_{v-1}(\mu)]\\+P(p\geq C_{v-1}(\mu)) C_{v-1}(\mu),\;~~~~~~\forall v\geq 1.
\end{multline}
The optimal policy results in a threshold price $C_v(\mu)$ for each battery energy level $v$, which is the maximum price the platform operator is willing to pay for one unit of energy for a vehicle with battery energy level $v$.

Let us denote the state of each EV using the tuple $(v,p)$.
Following the optimal policy, an EV with state $(v,p)$ takes an action to charge or travel, and transitions to a new state:
\begin{itemize}
    \item $(v+1,p)$ if charging.
    \item $(v-1,p')$ if traveling, with $p'$ sampled from $f_P(p)$.
\end{itemize}
Hence, the vehicle's charging decision process allows us to model the state of the vehicle as a Markov chain on the state space $(v,p)$, converging to a stationary distribution  $d(v,p)$\footnote{This Markov chain as we have defined it has a mix of continuous and discrete states, and it is straightforward to show that it satisfies the conditions for the existence of a unique stationary distribution. We remove the discussion for brevity and refer the reader to \cite{stationary}.  }. Note that the marginal distribution of prices observed under the stationary case is different from $f_P(p)$ as the vehicle might  charge for more than one unit  at a cheaper node and hence observing cheaper prices becomes more likely. Using the stationary distribution, the average charging cost per trip under the optimal policy $\mu$ can be written as:
\begin{equation}\label{pavg}
    p_{avg}(\mu) = \hspace{0pt}\sum_{v'=0}^{v_{\max}-1} \hspace{0pt}P_c(v = v') \mathbb{E}_{d(\cdot)}[p| p< C_{v'}(\mu), v=v'],
\end{equation}
where we define $P_c(v=v')$ as the probability that a charging vehicle has SoC $v'$. The average price paid for a charging vehicle with SoC $v'$ is calculated by the expected value of prices observed by the vehicle in the stationary distribution. Furthermore, we can explicitly write down the terms in \eqref{pavg}:
\begin{equation}\label{eq:pc}
\begin{aligned}
   P_c(v = v')&=\frac{\int_{p_{\min}}^{C_{v'}(\mu)}d(v',p)\;dp}{\sum_{v=0}^{v_{\max}-1}\int_{p_{\min}}^{C_{v}(\mu)}d(v,p)\;dp}\\&=\frac{\int_{p_{\min}}^{C_{v'}(\mu)}d(v',p)\;dp}{\frac{1}{1+\tau}},
\end{aligned}
\end{equation}
\begin{equation}\label{eq:expectedprice}
    \mathbb{E}_{d(\cdot)}[p| p< C_{v'}(\mu), v=v']=\frac{\int_{p_{\min}}^{C_{v'}(\mu)}p\;d(v',p)\;dp}{\int_{p_{\min}}^{C_{v'}(\mu)}d(v',p)\;dp}.
\end{equation}
The second equality in \eqref{eq:pc} follows from the fact that the trips take $\tau$ periods while charging takes one period. Hence, the probability that a vehicle is charging in the stationary distribution has to be $\frac{1}{1+\tau}$. Consequently, we get:
\begin{equation}
    \label{eq:avgcostrand}
    p_{avg}(\mu)=(1+\tau)\sum_{v'=0}^{v_{\max}-1}\int_{p_{\min}}^{C_{v'}(\mu)}p\; d(v',p)\; dp.
\end{equation}
For brevity of notation, we drop the dependence of the variables on $\mu$ from now on.
\subsection{Average Charging Cost}
In this section, we determine the average charging cost per vehicle $p_{avg}$ under the optimal charging strategy. In order to do this, first, we need to characterize the stationary distribution $d(v,p)$. At a given battery energy level $v$, $d(v,p)$ has to satisfy the following balance condition:
\begin{equation}
    \label{equilibrium}
    d(v,p)=d(v-1,p) u(C_{v-1}-p)+f_P(p)\int^{p_{\max}}_{C_{v+1}}d(v+1,p)\;dp,
\end{equation}
where $u(\cdot)$ is the unit-step function. The first term of summation corresponds to the vehicles that have made a charging decision at battery energy level $v-1$ and stayed at the same node, and the second term corresponds to the vehicles that have not charged at battery energy level $v+1$ and are randomly being distributed over prices after completing a trip.

In general, it is not possible to write down the average charging cost for any price distribution $f_P(p)$, because $d(v,p)$ can not be written in closed-form. To get closed-form results, we will make the following assumption:
\newtheorem{assumption}{Assumption}
\begin{assumption}
\label{ass:uniformprice}
The prices are uniformly distributed in $[p_{\min},p_{\max}]$. $f_P(p)=\frac{1}{p_{\min}-p_{\max}},\; p_{\min}\leq p\leq p_{\max}$.
\end{assumption}
In this case, $C_v$ is given by:
\begin{equation}
\label{Cvuniform}
    C_v=\eta-\frac{(p_{\max}-C_{v-1})^2}{2(p_{\max}-p_{\min})},\;\forall\;v\geq 2,
\end{equation}
with $C_1=\eta=\frac{p_{\min}+p_{\max}}{2}$. It is straightforward to go from \eqref{Cv} to \eqref{Cvuniform} through simple probabilistic calculations. When Assumption \ref{ass:uniformprice} holds, $d(v,p)$ becomes constant in the region of interest $[p_{\min},C_v]$. As a consequence, the integral in \eqref{eq:avgcostrand} can be calculated, and thus the average charging cost $p_{avg}$.
\begin{proposition}
\label{prop:chargecost}
When electricity prices follow Assumption \ref{ass:uniformprice},
\begin{equation}
    p_{avg}  = C_{v_{\max}}.
\end{equation}
\end{proposition}
According to the definition of $C_{v_{\max}}$, if we let a vehicle leave a node with energy level $v_{\max}$, it is going to pay an expected price of $C_{v_{\max}}$ the next time it charges. On the other hand, Proposition \ref{prop:chargecost} provides a stronger statement. If we let a vehicle with energy level $v_{\max}$ keep making trips and follow the optimal charging strategy, the average price paid for the electricity is still $C_{v_{\max}}$. As a result, it is rather straightforward to show: 1) The average charging cost $p_{avg}$ is a strictly decreasing function of $v_{\max}$; 2) As $v_{\max}$ goes to infinity, $p_{avg}$ goes to $p_{\min}$.
Seeing as $p_{avg}$ has these properties, the platform operator faces a trade-off between decreasing its charging costs by investing in vehicles with larger battery capacity or decreasing its investment and operational costs by operating vehicles with smaller batteries. This is the trade-off we study next.
\subsection{Trade-Off Between Operational and Charging Costs}
In this section, we propose an approach to choose the optimal battery capacity $v_{\max}$ and characterize $p_{avg}$ arising from this choice of battery capacity. To assign a cost to the choice of battery capacity, we make the following assumption:
\begin{assumption}
\label{batterycost}
The normalized (per period) cost of operating vehicles with battery capacity $v_{\max}$ is an affine function given by $\beta = \beta_0+ \xi v_{\max}$.
\end{assumption}
With this assumption, we are essentially breaking down the operational costs $\beta$   into two components: 1) $\beta_0$, a fixed cost to operate the vehicles which could represent mileage and maintenance costs,  and 2) $\xi v_{\max}$, the operational battery cost that will affect our choice of $v_{\max}$.
\begin{proposition}
\label{prop:batterycostfunction}
For  $\xi\leq\frac{p_{\max}-p_{\min}}{8}$, and given an optimal choice of battery capacity, $p_{avg}$ is:
\begin{equation}
\label{batterycostfunction}
\begin{aligned}
    &p_{avg}=\sqrt{2\xi(p_{\max}-p_{\min})}+p_{\min}. 
    \end{aligned}
\end{equation}
\end{proposition}
\begin{proof}
It is optimal to increase the battery capacity up to $v_{\max}$ such that $C_{v_{\max}}-C_{v_{\max}+1}=\xi$, because the marginal decrease in average charging cost resulting from increasing the battery capacity by one unit is canceled by the increased operational costs. Substituting $C_{v_{\max}+1}$ using \eqref{Cvuniform}:
\begin{equation*}
    \label{deltaC}
    \Delta C_{v_{\max}}=C_{v_{\max}}-C_{v_{\max}+1}=\frac{(C_{v_{\max}}-p_{\min})^2}{2(p_{\max}-p_{\min})}=\xi.
\end{equation*}
Rearranging the terms, we get \eqref{batterycostfunction}. Note that $v\geq1$, hence $C_{v_{\max}}\leq\frac{p_{\max}+p_{\min}}{2}$. This results in the constraint for $\xi$.
\end{proof}
The constraint on $\xi$ suggests that, if the cost of increasing the battery capacity for one unit is larger than $\frac{p_{\max}-p_{\min}}{8}$, then it is not beneficial to increase the battery capacity beyond $v_{\max}=1$, because the operational costs exceed the benefits. Note that Proposition \ref{prop:batterycostfunction} assumes that battery capacity units are small enough that we can always solve   $\Delta C_{v_{\max}}=\xi$. 

Proposition \ref{prop:batterycostfunction} gives us a direct relationship between battery cost and $p_{avg}$ under the random price model.  We can see that $p_{avg}$ is also dependent on $p_{\max}-p_{\min}$, which is a measure for the variance of the prices.
\begin{corollary}\label{prop:pricespread}
Let the mean of the prices be fixed at $\eta$. Then, as the standard deviation $\sigma$ of the prices increases, $p_{avg}$ decreases.
\end{corollary}

Corollary \ref{prop:pricespread} signifies the benefits of smart charging strategies. Specifically, when the deviation of the prices are higher around the same mean, the platform's optimal charging strategy results in even lower costs, as  the vehicles can {\it reap the benefits} of the lower prices that are more likely to be observed  while avoiding the higher spectrum of prices.

So far, we have assumed that vehicles can only decide to charge at the destination node of  their trips, and hence, we have not considered the possibility of rebalancing trips aiding the operator in decreasing the charging costs.
Unlike the network flow model, the random price model introduced so far does not allow the operator to perform rebalancing    to avoid charging at expensive nodes (as prices would be iid random after the rebalancing trip).  As such, we introduce a variation that allows the operator to perform rebalancing. Specifically, we
 assume that there exists a node $s$ outside of the network where the price for electricity $p_s$ is deterministic  and known to the operator, with $p_s \leq p_{\min}$. For ease of analysis, we consider the node to still be equidistant from all other nodes. This could represent a node equipped with renewable energy resources and storage devices, where cheap energy can be stored and later delivered to the vehicles.  Naturally, since  $p_s \leq p_{\min}$, any vehicle visiting $s$ charges to full. However, even though node $s$ provides cheap electricity, the rebalancing trips increase the operational costs (parameterized by $\beta$). Similar to the battery capacity-charging costs trade-off highlighted in Proposition \ref{prop:batterycostfunction}, a trade-off occurs between decreasing the charging costs and decreasing the operational costs due to higher number of trips. Given this modeling variation, our goal is to obtain the average charging cost and the additional rebalancing costs for traveling to node $s$ incurred by each vehicle under the optimal  strategy. In general it is not possible to write down   the average cost of charging and rebalancing per vehicle (denoted as $p_{avg}^r$)  in closed-form. 
 Nonetheless,  we generate approximate results.
\begin{proposition}\label{prop:rebalancingcost}
Let $p_{avg}$ be the average charging cost without rebalancing for a battery capacity $v_{max}$.
 For $v_{\max}\geq 3$, $p_{\min}\leq\ b \leq(2p_{avg}-p_{\min})$:
\begin{equation}
    \label{eq:optimalrebalancingcost}
    p_{avg}^r\approx b-\frac{(b-p_{\min})^2}{4(p_{avg}-p_{\min})},
\end{equation}
where $b=\frac{2}{v_{\max}-2}((1+\tau)\beta+p_s)+p_s$.
\end{proposition}
The constraint on $v_{\max}$ is to provide the appropriate setting for rebalancing, considering the rebalancing trip itself consumes 2 units 
energy.
The term $\frac{2}{v_{\max}-2}((1+\tau)\beta+p_s)$  illustrates the overhead the rebalancing trip causes in order to charge for one unit of energy at node s, and the additive $p_s$ is the price of electricity for one unit of energy. We refer to $b$ as the \textit{rebalancing cost}. In addition, $p_{avg}^r$ is always less than or equal to $b$, with equality if the constraint on $b$ achieves lower bound (in which case all the vehicles are sent for charging at node $s$). Moreover, the average cost with rebalancing is less than $p_{avg}$. Observe that the cost function is monotonically increasing in $b$, and equals to $p_{avg}$ at the upper bound. 

The amount of overhead caused is directly proportional to $\frac{1}{v_{\max}-2}$. Hence,  the costs are decreasing as $v_{\max}$ increases, which indicates the importance of the battery capacity in rebalancing as well. 

To conclude, although the results are approximate, they provide interesting insights on the benefits of rebalancing trips for charging.  In Section \ref{numerical}, we will study the quality of the approximations used in Proposition \ref{prop:rebalancingcost} numerically.
\begin{figure*}[t]
    \centering
    \includegraphics[width=0.9\textwidth]{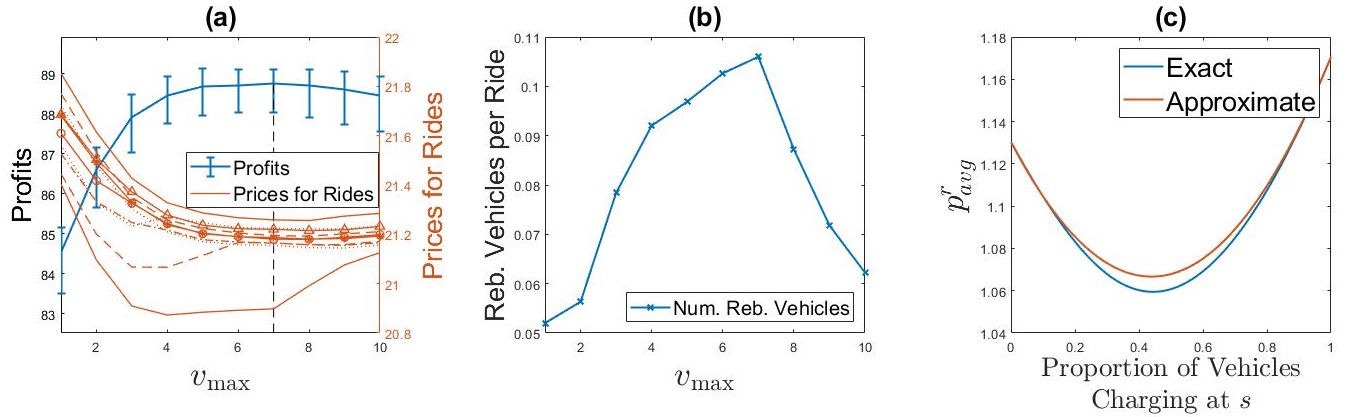}
    \caption{(a) Profits and prices for rides as a function of $v_{\max}$, (b) Number of rebalancing vehicles (normalized per trip) employed as a function of $v_{\max}$, (c) Comparison of approximate and exact solutions for average charging cost with rebalancing in the random price model for $v_{\max}=9$.}
    \label{fig:numerical}
\end{figure*}
\section{Numerical Results}\label{numerical}

In this section, we provide numerical results for the optimization problem in \eqref{eq:flowoptimization} and the approximation in Proposition \ref{prop:rebalancingcost}. For our analysis, we consider {\it one unit of energy} as described in the paper to be 10 kWh, and the operational battery cost per unit of energy, denoted as $\xi$, to be $\$0.003$ per period (normalized over 8 years)\cite{batterycost}. Moreover, each 6 minutes is considered as a discrete time unit \cite{chargeduration} (i.e., it takes the EVs 6 minutes to charge for 10kWh). The operational cost per period an EV is $\beta_0=\$0.1$ per period for a Tesla Model S\cite{teslacost}(normalized over 8 years). Price of electricity per unit of energy (10kWh) ranges from $\$0.8$ to $\$3$\cite{electricitycost}, $p_s=\$0.6$ and riders' maximum willingness to pay $\ell_{\max}=\$40$. The duration of trips is assumed to take $\tau=10$ time periods.

For the optimization problem in \eqref{eq:flowoptimization}, we use $m=10$ nodes with prices for electricity sampled from a uniform distribution in $[0.8,3]$. The problem was solved for 300 randomly created networks and the results were averaged. Figure \ref{fig:numerical}.a illustrates the profits and prices for rides originating from each node (The error bars indicate the maximum and the minimum profits out of 300 networks). Observe that the profits are increasing until $v_{max}=7$. However, since the marginal profits gained by increasing $v_{\max}$ by one are decreasing, investing in a battery capacity larger than $7$ causes the operational costs to dominate and hence the profits to decrease. Furthermore, the prices for rides display a decreasing behaviour as profits increase. However, it is interesting to note that the prices for rides at some nodes increase in the optimal solution as $v_{max}$ approaches its optimal value, because the globally optimal routing and charging strategy in the most general network flow model is different for each $v_{max}$, and hence might increase the costs of rides originating from certain nodes.

Figure \ref{fig:numerical}.b highlights the importance of rebalancing vehicles. As $v_{\max}$ increases up to its maximum of 7, the number of rebalancing vehicles employed per ride increases. Even though this increases the operational costs, their use for charging purposes  decreases the total costs and thus increases the profits. As $v_{\max}$ increases beyond optimum, growing operational costs result in less rebalancing vehicles employed.

Finally, in Figure \ref{fig:numerical}.c we plot the total average cost per vehicle versus the proportion of vehicles charging at node $s$. With optimal rebalancing, the average costs can be reduced substantially (from 1.13 to 1.06, around $6\%$). Observe that the approximate solution given by Proposition \ref{prop:rebalancingcost} and exact solution show very little error, which displays the fairness of our approximation.
\section{Conclusion}
In this paper, we presented the benefits of smart charging in an AMoD fleet of EVs controlled by a profit-maximizing platform operator. By first showing that the profits generated are highly dependent on the charging and operational costs the rides incur, we proposed a smart charging strategy in order to minimize these costs. We show that investing in a larger battery and utilizing more vehicles for rebalancing decrease the charging costs. However, due to the diminishing returns and increasing operational costs, there exists an optimal number of vehicles to operate for rebalancing and an optimal battery capacity to invest in. Aside from the numerical studies that support our claims, we provided  closed-form  expressions  for  the  average charging cost under optimal investment decisions, which we believe provide insights for design specifications and operating strategies that are crucial in an EV AMoD system.

\bibliographystyle{IEEEtran}
\bibliography{references}

\begin{thebibliography}{10}
\providecommand{\url}[1]{#1}
\csname url@samestyle\endcsname
\providecommand{\newblock}{\relax}
\providecommand{\bibinfo}[2]{#2}
\providecommand{\BIBentrySTDinterwordspacing}{\spaceskip=0pt\relax}
\providecommand{\BIBentryALTinterwordstretchfactor}{4}
\providecommand{\BIBentryALTinterwordspacing}{\spaceskip=\fontdimen2\font plus
\BIBentryALTinterwordstretchfactor\fontdimen3\font minus
  \fontdimen4\font\relax}
\providecommand{\BIBforeignlanguage}[2]{{%
\expandafter\ifx\csname l@#1\endcsname\relax
\typeout{** WARNING: IEEEtran.bst: No hyphenation pattern has been}%
\typeout{** loaded for the language `#1'. Using the pattern for}%
\typeout{** the default language instead.}%
\else
\language=\csname l@#1\endcsname
\fi
#2}}
\providecommand{\BIBdecl}{\relax}
\BIBdecl

\bibitem{companies}
[Online]. Available: https://www.cbinsights.com/research/autonomous-
  driverless-vehicles-corporations-list/.

\bibitem{simulation5}
L.~M. Martinez and J.~M. Viegas, ``{Assessing the impacts of deploying a shared
  self-driving urban mobility system: An agent-based model applied to the city
  of Lisbon, Portugal},'' \emph{International Journal of Transportation Science
  and Technology}, vol.~6, no.~1, pp. 13--27, 2017.

\bibitem{simulation4}
C.~Ruch, S.~H\"orl, and E.Frazzoli, ``{AMoDeus, a Simulation-Based Testbed for
  Autonomous Mobility-on-Demand Systems},'' \emph{\textnormal{In} Proc. IEEE
  Int. Conf. on Intelligent Transportation Systems}, 2018.

\bibitem{simulation2}
M.~W. Levin, K.~M. Kockelman, S.~D. Boyles, and T.~Li, ``{A general framework
  for modeling shared autonomous vehicles with dynamic network-loading and
  dynamic ride-sharing application},'' \emph{Computers, Environment and Urban
  Systems}, vol.~64, pp. 373--383, 2017.

\bibitem{simulation1}
S.~H\"orl, C.~Ruch, F.~Becker, E.~Frazzoli, and K.~W. Axhausen, ``{Fleet
  control algorithms for automated mobility: A Simulation assessment for
  Zurich},'' \emph{\textnormal{In} TRB Annual Meeting}, 2018.

\bibitem{simulation6}
J.~Wen, Y.~X. Chen, N.~Nassir, and J.~Zhao, ``{Transit-oriented autonomous
  vehicle operation with integrated demand-supply interaction},''
  \emph{Transportation Research Part C: Emerging Technologies}, vol.~97, pp.
  216--234, 2018.

\bibitem{publictransit1}
M.~Salazar, F.~Rossi, M.~Schiffer, C.~H. Onder, and M.~Pavone, ``{On the
  interaction between autonomous mobility-on-demand and the public
  transportation systems},'' \emph{\textnormal{In} Proc. IEEE Int. Conf. on
  Intelligent Transportation Systems}, 2018.

\bibitem{queuetheoretical2}
R.~Zhang and M.~Pavone, ``{Control of robotic Mobility-on-Demand systems: A
  queueing-theoretical perspective},'' \emph{\textnormal{In} Int. Journal of
  Robotics Research}, vol.~35, no. 1--3, pp. 186--203, 2016.

\bibitem{networkflow1}
M.~Pavone, S.~L. Smith, E.~Frazzoli, and D.~Rus, ``{Robotic load balancing for
  Mobility-on-Demand systems},'' \emph{Int. Journal of Robotics Research},
  vol.~31, no.~7, pp. 839--854, 2012.

\bibitem{recedinghorizon1}
R.~Iglesias, F.~Rossi, K.~Wang, D.~Hallac, J.~Leskovec, and M.~Pavone,
  ``{Data-driven model predictive control of autonomous mobility-on-demand
  systems},'' \emph{\textnormal{In} Proc. IEEE Conf. on Robot. and Autom.},
  2018.

\bibitem{congestion1}
F.~Rossi, R.~Zhang, Y.~Hindy, and M.~Pavone, ``{Routing autonomous vehicles in
  congested transportation networks: Structural properties and coordination
  algorithms},'' \emph{Autonomous Robots}, vol.~42, no.~7, pp. 1427--1442,
  2018.

\bibitem{evtaxi1}
J.~Bischoff and M.~Maciejewski, ``{Agent-based simulation of electric taxicab
  fleets},'' \emph{Transp. Res. Procedia}, vol.~4, pp. 191--198, 2014.

\bibitem{evtaxi2}
H.~Wang and R.~Cheu, ``{Operations of a taxi fleet for advance reservations
  using electric vehicles and charging stations},'' \emph{Transp. Res. Rec.: J.
  of the Transp. Res. Board}, vol. 2352, no.~1, pp. 1--10, 2013.

\bibitem{scheduling1}
I.~Gkourtzounis, E.~Rigas, and N.~Bassiliades, ``{Towards Online Electric
  Vehicle Scheduling for Mobility-On-Demand Schemes},'' \emph{\textnormal{In:
  Slavkovik M. (eds) Multi-Agent Systems. EUMAS 2018. Lecture Notes in Computer
  Science, vol 11450. Springer, Cham}}, 2016.

\bibitem{inroutecharging}
M.~Ammous, S.~Belakaria, S.~Sorour, and A.~Abdel-Rahim, ``{Optimal Routing with
  In-Route Charging of Mobility-on-Demand Electric Vehicles},''
  \emph{\textnormal{In} VTC-Fall}, 2017.

\bibitem{agentbased}
T.~D. Chen, K.~M. Kockelman, and J.~P. Hanna, ``{Operations of a Shared,
  Autonomous, Electric Vehicle Fleet: Implications of Vehicle \& Charging
  Infrastructure Decisions},'' \emph{Transportation Research Part A: Policy and
  and Practice}, vol.~94, pp. 243--254, 2016.

\bibitem{nate}
N.~Tucker, B.~Turan, and M.~Alizadeh, ``{Online Charge Scheduling for Electric
  Vehicles in Autonomous Mobility on Demand Fleets},'' \emph{\textnormal{In}
  Proc. IEEE Int. Conf. on Intelligent Transportation Systems}, 2019.

\bibitem{fogbased}
S.~Belakaria, M.~Ammous, S.~Sorour, and A.~Abdel-Rahim, ``{Fog-Based
  Multi-Class Dispatching and Charging for Autonomous Electric Mobility
  On-Demand},'' \emph{IEEE Transactions on Intelligent Transportation Systems},
  pp. 1--15, 2019.

\bibitem{trainconnection}
X.~Liang, G.~H. de~A.~Correia, and B.~van Arem, ``{Optimizing the service area
  and trip selection of an electric automated taxi system used for the last
  mile of train trips},'' \emph{Transportation Research Part E}, vol.~93, pp.
  115--129, 2016.

\bibitem{powergrid1}
F.~Rossi, R.~Iglesias, M.~Alizadeh, and M.~Pavone, ``{On the interaction
  between Autonomous Mobility-on-Demand systems and the power network: Models
  and coordination algorithms},'' \emph{\textnormal{In} Robotics: Science and
  Systems}, 2018.

\bibitem{ridesharing}
K.~Bimpikis, O.~Candogan, and D.~Saban, ``{Spatial Pricing in Ride-Sharing
  Networks},'' 2016, available at SSRN: https://ssrn.com/abstract=2868080.

\bibitem{stationary}
{Markov Chains on Continuous State Space}. {https://
  www.webpages.uidaho.edu/$\sim$stevel/565/lectures/5d\%20MCMC.pdf}.

\bibitem{batterycost}
(2018) {Electric Vehicle Battery: Materials, Cost, Lifespan}. [Online].
  Available:
  https://www.ucsusa.org/clean-vehicles/electric-vehicles/electric-cars-battery-life-materials-cost.

\bibitem{chargeduration}
\BIBentryALTinterwordspacing
{How Long Does It Take to Charge an Electric Car?} [Online]. Available:
  \url{https://pod-point.com/guides/driver/how-long-to-charge-an-electric-car}
\BIBentrySTDinterwordspacing

\bibitem{teslacost}
\BIBentryALTinterwordspacing
 [Online]. Available: \url{www.tesla.com}
\BIBentrySTDinterwordspacing

\bibitem{electricitycost}
Edmunds. (2019) {The True Cost of Powering an Electric Car}. [Online].
  Available:
  https://www.edmunds.com/fuel-economy/the-true-cost-of-powering-an-electric-car.html.

\end{thebibliography}

\appendix
\subsection{Proof of Proposition \ref{prop:networkflow}}
\begin{proof}
Let $\nu_i^v$ be dual variables corresponding to the equilibrium flow constraints and $\lambda_{ij}$ be dual variables corresponding to the demand satisfaction constraints. 
Since the optimization problem   \eqref{eq:flowoptimization} is a convex quadratic maximization problem (given a with uniform $F(\cdot)$) and Slater's condition is satisfied, strong duality holds. We can write the dual problem as:
\begin{equation}
    \label{eq:dual}
    \begin{aligned}
    &\underset{\lambda_{ij},\nu_i^v}{\text{min}}\underset{\ell_i}{\text{max}}
    & &\sum_{i=1}^m\left(\theta_i(1-\frac{\ell_i}{\ell_{\max}})\left(\ell_i-\sum_{j=1}^m\lambda_{ij}\alpha_{ij}\right)\right)\\
    & \text{subject to}
    & & \nu_i^v-\nu_i^{v+1}-p_i-\beta\leq 0,\\
    & & & \lambda_{ij}+\nu_i^v-\nu_j^{v-1}-\tau\beta\leq 0,\\
    & & & \nu_i^v-\nu_j^{v-1}-\tau\beta\leq 0\quad \forall i,j,v.
    \end{aligned}
\end{equation}
For fixed $\lambda_{ij}$ and $\nu_i^v$, the inner maximization results in the optimal prices:
\begin{equation}
\label{eq:optimalpricesproof}
    \ell_i^*=\frac{\ell_{\max}+\sum_{j=1}^m\lambda_{ij}\alpha_{ij}}{2}.
\end{equation}
By strong duality, the optimal primal solution satisfies the dual solution with optimal dual variables $\lambda_{ij}^*$, ${\nu_i^v}^*$, which completes the first part of the proposition. The dual problem with optimal prices in \eqref{eq:optimalpricesproof} can be written as:
\begin{equation}
    \label{eq:dualwithoptimalprices}
    \begin{aligned}
    &\underset{\lambda_{ij},\nu_i^v}{\text{min}}
    & &\sum_{i=1}^m\frac{\theta_i}{\ell_{\max}}\left(\frac{\ell_{\max}-\sum_{j=1}^m\lambda_{ij}\alpha_{ij}}{2}\right)^2\\
    & \text{subject to}
    & & \nu_i^v-\nu_i^{v+1}-p_i-\beta\leq 0,\\
    & & & \lambda_{ij}+\nu_i^v-\nu_j^{v-1}-\tau\beta\leq 0,\\
    & & & \nu_i^v-\nu_j^{v-1}-\tau\beta\leq 0\quad \forall i,j,v.
    \end{aligned}
\end{equation}
The objective function in \eqref{eq:dualwithoptimalprices} with optimal dual variables, along with \eqref{eq:optimalprices} suggests:
\begin{equation*}
    P=\sum_{i=1}^m\frac{\theta_i}{\ell_{\max}}(\ell_{\max}-\ell_i^*)^2,
\end{equation*}
where profits $P$ is the value of the objective function of both optimal and dual problems.
To upper bound $\ell_i^*$, we upper bound $\lambda_{ij}$ using the constraints in \eqref{eq:dualwithoptimalprices}. The proof is straightforward through algebraic calculations.
\end{proof}
\subsection{Proof of Proposition \ref{prop:chargecost}}
For convenience and ease of notation in our calculations, we are going to assume $C_0=p_{max}$. This does not violate our model, as $d(0,p)$ is zero for $p>p_{\max}$.
To prove Proposition \ref{prop:chargecost}, we first state the following lemmas:
\begin{lemma}
\label{lem:densities}
For $v=0$, $d(0,p)$ is constant in $[p_{\min},p_{\max}]$. For any other $v$, $d(v,p)$ is a constant denoted as $d_v^1$ in the interval $[p_{\min},C_{v-1})$, and another constant $d_v^2$ in $[C_{v-1},p_{\max}]$. 
\end{lemma}
 
The lemma simply follows from the balance condition \eqref{equilibrium}. We exclude the proof for brevity.

 When we have the proposed characteristic of the stationary distribution, the recursive relation between $d(v,p)$ and $d_(v-1,p)$ can be written as follows based on \eqref{equilibrium}:\\
\begin{equation}
    \label{densitymatrix}
    \renewcommand*{\arraystretch}{2}
\begin{aligned}
   & \left[ \begin{array}{c} d_v^1 \\ d_v^2 \end{array} \right] = \begin{bmatrix} \frac{p_{\max}-C_{v-1}}{p_{\max}-C_v} & \frac{p_{\max}-p_{\min}}{p_{\max}-C_v} \\ \frac{C_v-C_{v-1}}{p_{\max}-C_v} & \frac{p_{\max}-p_{\min}}{p_{\max}-C_v} \end{bmatrix} \left[ \begin{array}{c} d_{v-1}^1 \\ d_{v-1}^2 \end{array} \right],\\
&    \forall v\geq 1\;\textnormal{and}\; d_0^1=d_0^2=d_0.
\end{aligned}
\end{equation}
To get the average charging cost, we are solely interested in the values of the $d_v^1$'s, because charging is done in the range $[p_{\min},C_v)$. Another way of writing a relation between the stationary distributions is as follows:
\begin{lemma}
\label{recursive}
$d_v^1=\frac{p_{\max}-p_{\min}}{p_{\max}-C_{v}} d_{v-1}^1$.
\end{lemma}
\begin{proof}
Equation \eqref{densitymatrix} yields $d_{v-1}^2=d_{v-1}^1-d_{v-2}^1$. Hence, by substituting $d_{v-1}^2$ with this, we get:
\begin{multline}
\label{recursivedensity}
    d_v^1=\frac{1}{p_{\max}-C_v}[(p_{\max}-C_{v-1})d_{v-1}^1\\+(p_{\max}-p_{\min})d_{v-1}^1-(p_{\max}-p_{\min})d_{v-2}^1].
\end{multline}
 By induction, if Lemma \ref{recursive} holds for $d_{v-1}^1$, then it holds for $d_v^1$, since the first and the last terms inside square brackets cancel each other. For $v=1$, the lemma holds since \eqref{recursivedensity} gives $(p_{\max}-C_1)d_1^1=(p_{\max}-p_{\min})d_0^1$, because $d_0^1=d_0^2$.
\end{proof}
Finally, we can prove Proposition \ref{prop:chargecost}:
\begin{proof}
Using Lemma \ref{recursive}, for any $v$, $d_v^1$ is given by:
\begin{equation}
\label{dv1}
    d_v^1=d_0 \prod_{i=1}^v\frac{p_{\max}-p_{\min}}{p_{\max}-C_i}.
\end{equation}
This allows us to write $d_v^1$'s and $d_v^2$'s in terms of $d_0$. Moreover, since $d(v,p)$ is the distribution of a vehicle, integration over $p$ and summation over $v$ should be equal to 1. This normalization yields:
\begin{equation}
\label{d01}
    d_0^1=d_0^2=d_0=\frac{\prod_{i=1}^{v_{\max}-1}(p_{\max}-C_i)}{(1+\tau)(p_{\max}-p_{\min})^{v_{\max}}}.
\end{equation}
The final step to get the average charging cost per vehicle is to use \eqref{eq:avgcostrand} for uniformly distributed prices. The equation takes the following form:
\begin{equation}
    \label{averagecostuniform}
    p_{avg}=(1+\tau)\sum_{v=0}^{v_{\max}-1}d_v^1\frac{C_v^2-p_{\min}^2}{2}.
\end{equation}
Substituting \eqref{dv1} and \eqref{d01} into \eqref{averagecostuniform}:
\begin{equation}
\label{uniformchargecostcomplete}
p_{avg}=\sum_{v=0}^{v_{max-1}}\frac{C_v^2-p_{\min}^2}{2}\cdot\frac{\prod_{i=v+1}^{v_{\max}-1}(p_{\max}-C_i)}{(p_{\max}-p_{\min})^{v_{\max}-v}}.
\end{equation}
If we write down \eqref{Cv} explicitly for uniform distribution and $v=v_{\max}$, we get:
\begin{equation*}
\label{Cvmaxexplicit}
    \begin{aligned}
       C_{v_{\max}}
       &=\frac{C_{v_{\max}-1}^2-p_{\min}^2}{2}\cdot\frac{1}{p_{\max}-p_{\min}}\\
       &\;\;\;\;+\frac{p_{\max}-C_{v_{\max}-1}}{p_{\max}-p_{\min}}\cdot \left[C_{v_{\max}-1}\right].
    \end{aligned}
\end{equation*}
Next, by writing $C_{v_{\max}-1}$ in curly brackets explicitly in terms of $C_{v_{\max}-2}$, and then further applying this method until $C_0$, we get the same expression as in \eqref{uniformchargecostcomplete}.
\end{proof}
\subsection{Proof of Corollary \ref{prop:pricespread}}
The proof simply follows from writing $p_{avg}$ in \eqref{batterycostfunction} in terms of $\eta$ and $\sigma$, taking derivative with respective to $\sigma$, and applying the upper bound on $\xi$. We omit the proof for brevity.
\subsection{Proof of Proposition \ref{prop:rebalancingcost}}
First, we need to prove the following lemma:
\begin{lemma}
The rebalancing trips are only made by vehicles with battery energy level $v=1$.
\end{lemma}
\begin{proof}
It is optimal for the platform operator to first send vehicles with state $v=1$ for rebalancing, as they would end up paying the highest expected price for electricity at $v=0$ state. Let $\gamma$ be the portion of the vehicles with state $v=1$ sent for rebalancing. The expected cost of the next charge is modified as:
\begin{equation}\label{eq:gammaCv}
C_1(\gamma)=(1-\gamma)\eta+\gamma C_{v_{\max}}(\gamma)
\end{equation}
When $\gamma=1$, $C_v=p_{\min}$ for all $v$ as the recursion in \eqref{Cvuniform} suggests. Hence, no vehicle charges at regular nodes.
\end{proof}
In general, we can not characterize $C_{v_{\max}}(\gamma)$ in closed-form. Instead, we consider the proportion of the charging vehicles remaining at regular nodes (excluding the overhead caused by rebalancing), denoted by $n$, and approximate the behaviour of $C_{v_{\max}}$ linearly with n:
\begin{equation}
    \label{eq:cvmaxapproximate}
    C_{v_{\max}}(n)\approx p_{\min}+n(p_{avg}-p_{\min}).
\end{equation}
Note that $\gamma=1$ corresponds to $n=0$ and $\gamma=0$ corresponds to $n=1$, hence the endpoints are satisfied.
The total rebalancing and charging costs the platform incurs:
\begin{equation}    \label{eq:costapproximate}
    p_{avg}^r\approx n(p_{\min}+n(p_{avg}-p_{\min}))+(1-n)b,
\end{equation}
where $b=\frac{2}{v_{\max}-2}((1+\tau)\beta+p_s)+p_s$. Since 2 units of energy is wasted for the trips caused by the rebalancing vehicles, $\frac{2}{v_{\max}-2}((1+\tau)\beta+p_s)$ corresponds to the overhead induced by rebalancing. Specifically, to provide $v_{\max}-2$ units of energy at node $s$, an excess cost of $2p_s$ to charge for the rebalancing trips and $(2+2\tau)\beta$ to operate the charging and traveling vehicles has to be paid.

The total average cost is a weighted average of charging costs at regular nodes and cost of rebalancing. The minimum of $p_{avg}^r$ follows from minimizing \eqref{eq:cvmaxapproximate} with respect to $n$.
\end{document}